\documentclass[12pt,reqno]{article}

\usepackage[utf8]{inputenc}
\usepackage[english]{babel}
\usepackage[usenames]{color}
\usepackage{amssymb}
\usepackage{graphicx}
\usepackage{amscd}

\usepackage[colorlinks=true,
linkcolor=webgreen,
filecolor=webbrown,
citecolor=webgreen]{hyperref}

\definecolor{webgreen}{rgb}{0,.5,0}
\definecolor{webbrown}{rgb}{.6,0,0}

\usepackage{color}
\usepackage{fullpage}
\usepackage{float}

\usepackage{graphics,amsmath,amssymb}
\usepackage{amsthm}
\usepackage{amsfonts}
\usepackage{latexsym}
\usepackage{epsf}
\usepackage{lineno}

\begin{document}

\begin{center}
\epsfxsize=4in
\end{center}

\theoremstyle{plain}
\newtheorem{theorem}{Theorem}
\newtheorem{corollary}[theorem]{Corollary}
\newtheorem{lemma}[theorem]{Lemma}
\newtheorem{proposition}[theorem]{Proposition}

\theoremstyle{definition}
\newtheorem{definition}[theorem]{Definition}
\newtheorem{example}[theorem]{Example}
\newtheorem{conjecture}[theorem]{Conjecture}

\theoremstyle{remark}
\newtheorem{remark}[theorem]{Remark}

\def\R{{\mathbb R}}
\def\Z{{\mathbb Z}}
\def\H{{\mathbb H}}
\def\C{{\mathbb C}}
\def\N{{\mathbb N}}
\def\G{{\mathbb G}}
\def\S{{\mathbb S}}
\def\F{{\mathbb F}}
\def\K{{\mathbb K}}
\def\QQ{{\mathbb Q}}
\def\cA{{\mathcal A}}
\def\cB{{\mathcal B}}
\def\cC{{\mathcal C}}
\def\cD{{\mathcal D}}
\def\cF{{\mathcal F}}
\def\cG{{\mathcal G}}
\def\cH{{\mathcal H}}
\def\cK{{\mathcal K}}
\def\cL{{\mathcal L}}
\def\cM{{\mathcal M}}
\def\cN{{\mathcal N}}
\def\cP{{\mathcal P}}
\def\cR{{\mathcal R}}
\def\cS{{\mathcal S}}
\def\cT{{\mathcal T}}

\def\bw{\mathbf{w}}
\def\bx{\mathbf{x}}
\def\by{\mathbf{y}}

\def \Q {{\bf Q}}
\def \RP {{\bf RP}}
\def \id {{\rm id}}
\def \e {\epsilon}
\def\f{f}
\def\card{\mathrm{Card}}


\begin{center}
\vskip 1cm{\LARGE\bf The constant of recognizability is computable 
\vskip .12in
for primitive morphisms}
\vskip 1cm
\large
Fabien Durand \\
Laboratoire Ami\'enois de Math\'ematiques Fondamentales et Appliqu\'ees, CNRS-UMR 7352, \\
Universit\'{e} de Picardie Jules Verne, \\
33 rue Saint Leu, 80000 Amiens,\\
France\\
\href{mailto:fabien.durand@u-picardie.fr}{\tt fabien.durand@u-picardie.fr} \\
\ \\
Julien Leroy \\
Institut de math\'ematique, \\
Universit\'e de Li\`ege \\
All\'ee de la d\'ecouverte 12 (B37) \\
4000 Li\`ege, \\
Belgium \\
\href{J.Leroy@ulg.ac.be}{\tt J.Leroy@ulg.ac.be}\\
\end{center}

\vskip .2 in

\begin{abstract}
Mossé proved that primitive morphisms are recognizable. In this paper we give a computable upper bound for the constant of recognizability of such a morphism.
This bound can be expressed only using the cardinality of the alphabet and the length of the longest image under the morphism of a letter.
\end{abstract}

\section{Introduction}

Infinite words, i.e., infinite sequences of symbols from a finite set, usually called alphabet, 
form a classical object of study. 
They have an important representation power: they provide a natural way to code elements of an
infinite set using finitely many symbols, e.g., the coding of an orbit in a discrete dynamical system or 
the characteristic sequence of a set of integers. 
A rich family of infinite words, with a simple algorithmic description, 
is made of the words obtained by iterating a morphism $\sigma:A^* \to A^*$ \cite{Choffrut&Karhumaki:1997}, where $A^*$ is the free monoid generated by the finite alphabet $A$.

If $\sigma$ is {\em prolongable} on some letter $a \in A$, that is, if $\sigma(a) = au$ for some non-empty word $u$ and $\lim_{n \to +\infty} |\sigma^n(a)| = +\infty$, then $\sigma^n(a)$ converges to an infinite word $x = \sigma^{\omega}(a) \in A^\mathbb{N}$ that is a fixed point of $\sigma$.
Two-sided fixed points are similarly defined as infinite words of the form $\sigma^{\omega}(a \cdot b) \in A^\Z$, where $\sigma(a) = ua$ and $\sigma(b) = bv$ with $u,v \in A^+$ and $\lim_{n \to +\infty} |\sigma^n(a)| = \lim_{n \to +\infty} |\sigma^n(b)| = +\infty$.
Such a fixed point is said to be {\em admissible} if $ab$ occurs in $\sigma^n(c)$ for some $n \in \N$ and some $c \in A$.
When the morphism is {\em primitive}, i.e., there exists $k \in \mathbb{N}$ such that $b$ occurs in $\sigma^k(c)$ for all $b,c\in A$, then $x$ is {\em uniformly recurrent}: any finite word that occurs in $x$ occurs infinitely many times in it and with bounded gaps \cite{Queffelec:2010}.
The converse almost holds true: if $x=\sigma^{\omega}(a)$ is uniformly recurrent, then there exist a primitive morphism $\varphi:B^* \to B^*$, a letter $b \in B$ and a morphism $\psi:B^* \to A^*$ such that $x = \psi(\varphi^{\omega}(b))$ \cite{Durand:1998}.
We let $\cL(x)$ denote the set of {\em factors} of $x$, i.e., $\cL(x) = \{u \in A^* \mid \exists p \in A^*, w \in A^\N: x = puw\}$ (with a similar definition of two-sided fixed points). 

{\em Recognizability} is a central notion when dealing with fixed point of morphisms.
It is linked to existence of long powers $u^k$ in $\cL(x)$ \cite{Mignosi&Seebold:1993}.
An infinite word $x \in A^\Z$ is said to be {\em $k$-power-free} if there is no non-empty word $u$ such that $u^k$ belongs to $\cL(x)$.
We refer, for example, to \cite{Durand&Host&Skau:1999,Holton&Zamboni:1999,Canterini&Siegel:2001,Pytheas}.
It roughly means that any long enough finite word that occurs in $\sigma^{\omega}(a)$ has a unique pre-image under $\sigma$, except for a prefix and a suffix of bounded length which is called the {\em constant of recognizability}.
A fundamental result concerning recognizability is due to Mossé who proved that aperiodic primitive morphisms (i.e., primitive morphisms with aperiodic fixed points) are recognizable~\cite{Mosse:1992,Mosse:1996}.
In this paper, we present a detailed proof of this result. 
This allows us to give a bound on the constant of recognizability.

\section{Recognizability}

Given a morphism $\sigma:A^* \to A^*$, we respectively define $|\sigma|$ and $\langle \sigma \rangle$ by
\[
	|\sigma| = \max_{a \in A} |\sigma(a)|,
	\quad \text{and,} \quad
	\langle\sigma\rangle = \min_{a \in A} |\sigma(a)|.
\]
Assuming that $\sigma$ has an admissible fixed point $x \in A^\Z$, for all $p \in \N$, we let $f_x^{(p)}$ denote the function
$$
f_x^{(p)}: \Z \to \Z, i \mapsto 
f_x^{(p)} (i) = 
\begin{cases}
	|\sigma^p (x_{[0,i[})| 	& 	\text{if } i > 0,	\\
	0 						&	\text{if } i=0,		\\
	|\sigma^p (x_{[i,0[})| 	& 	\text{if } i < 0.
\end{cases}
$$
We set $E (x , \sigma^p) = f_x^{(p)}(\mathbb{Z})$.
When it is clear from the context, we simply write $f^{(p)}$ instead of $f_x^{(p)}$.

Given two integers $i,j$ with $i \leq j$, we let $x_{[i,j]}$ and $x_{[i,j[}$ respectively denote the factors $x_i x_{i+1} \cdots x_j$ and $x_i x_{i+1} \cdots x_{j-1}$ (with $x_{[i,i[} = \varepsilon$, where $\epsilon$ is the empty word, i.e., the neutral element of $A^*$).

\begin{definition}
We say that $\sigma$ is {\em recognizable on $x$} if there exists some constant $L > 0$ such that for all $i,m \in \Z$,
\[
(x_{[m-L,m+L]} =  x_{[f^{(1)}(i)-L,f^{(1)}(i)+L]}) 
\Rightarrow 
(\exists j \in \Z)((m = f^{(1)}(j)) \wedge (x_i =x_j)).  
\]
The smallest $L$ satisfying this condition is called the {\em constant of recognizability} of~$\sigma$ for $x$.
When $\sigma$ is recognizable on all its admissible fixed points, we say that it is recognizable and its constant of recognizability is the greatest one.
\end{definition}

\begin{lemma}
If $\sigma:A^* \to A^*$ is recognizable on the admissible fixed point $x \in A^\Z$ and if $L$ is the constant of recognizability of $\sigma$ for $x$, then for all $k>0$, $x$ is also an admissible fixed point of $\sigma^k$ and $\sigma^k$ is recognizable on $x$ and its constant of recognizability for $x$ is at most $L \frac{|\sigma|^k-1}{|\sigma|-1}$.
\end{lemma}
\begin{proof}
The result holds by induction on $k>0$.
The infinite word $x$ is obviously an admissible fixed point of $\sigma^k$.
With $L' = L \frac{|\sigma|^k-1}{|\sigma|-1}$, let us show that for all $i \in \Z$, the word 
\[
	x_{[f^{(k)}(i)-L',f^{(k)}(i)+L']}
\]
uniquely determines the letter $x_i$.

By recognizability, the word $x_{[f^{(k)}(i)-L',f^{(k)}(i)+L']}$ uniquely determines the word $x_{[m,M]}$, where $m$ is the smallest integer such that $f^{(k)}(i)-L' \leq f^{(1)}(m)-L$ and $M$ is the largest integer such that $f^{(1)}(M)+L \leq f^{(k)}(i)+L'$.
Therefore, the word $x_{[f^{(k)}(i)-L',f^{(k)}(i)+L']}$ uniquely determines the word 
\[
	x_{[f^{(k-1)}(i)-\frac{L'-L}{|\sigma|},f^{(k-1)}(i)+\frac{L'-L}{|\sigma|}]} = x_{[f^{(k-1)}(i)-L'',f^{(k-1)}(i)+L'']},
\]
where $L'' = L \frac{|\sigma|^{k-1}-1}{|\sigma|-1}$.
\end{proof}

\begin{theorem}
\label{theo:mosse}
Let $\sigma :A^* \to A^*$ be an aperiodic primitive morphism and let $x \in \Z$ be an admissible fixed point of $\sigma$.
\begin{enumerate}
\item
\cite{Mosse:1992}
There exists $M>0$ such that, for all $i,m \in \Z$,
$$
	x_{[f^{(1)}(i)-M,f^{(1)}(i)+M]} =  x_{[m-M,m+M]} \Longrightarrow m \in E (x, \sigma).  
$$
\item
\cite{Mosse:1996}
There exists $L>0$ such that, for all $i,j \in \Z$, 
$$
x_{[f^{(1)}(i)-L,f^{(1)}(i)+L]} =  x_{[f^{(1)}(j)-L,f^{(1)}(j)+L]} \Longrightarrow x_i = x_j.  
$$
\end{enumerate}
\end{theorem}

By a careful reading of the proofs of Moss\'e's results, we can improve it as follows. 
The proof is given in Section~\ref{section:detailed result}.
For an infinite word $x \in A^\Z$, we let $p_x:\N \to \N$ denote the {\em complexity function} of $x$ defined by $p_x(n) = \#\cL_n(x)$ where $\cL_n(x) = (\cL(x) \cap A^n)$.

\begin{theorem}
\label{theo:maindetail}
Let $\sigma :A^* \to A^*$ be a morphism with an admissible fixed point $x \in A^\Z$.
If $x$ is $k$-power-free and if there is some constant $N$ such that for all $n \in \mathbb{N}$, $|\sigma^n| \leq N \langle \sigma^n \rangle$, then $\sigma$ is recognizable on $x$ and its constant of recognizability for $x$ is at most $R |\sigma^{dQ}|+|\sigma^{d}|$, where
\begin{itemize}
\item
$R = \lceil  N^2(k+1)+2N\rceil$;
\item
$Q = 1 + p_x(R)\left(\sum_{\frac{R}{N} \leq i \leq RN+2} p_x(i) \right)$;
\item
$d \in \{1,2,\dots,\# A\}$ is such that for any words $u,v \in \cL(x)$, we have 
\[
	\sigma^{d-1}(u) \neq \sigma^{d-1}(v) \Rightarrow \forall n, \sigma^n(u) \neq \sigma^n(v). 
\] 
\end{itemize}
\end{theorem}

Then, we give some computable bounds for $N$, $R$, $k$, $Q$ and $d$ in the case of primitive morphisms.
These bounds are not sharp but can be expressed only using the cardinality of the alphabet and the maximal length $|\sigma|$.
The proof is given in Section~\ref{section: primitive case}.

\begin{theorem}
\label{theo:main}
Any aperiodic primitive morphism $\sigma :A^* \to A^*$ that admits a fixed point $x \in A^\Z$ is recognizable on $x$ and the constant of recognizability for $x$ is at most
\[
	2 |\sigma|^{6(\#A)^2 + 6(\#A) |\sigma|^{28(\#A)^2}} + |\sigma|^{(\#A)}. 
\]
\end{theorem}

The bound given in the previous theorem is far from being sharp. 
When the morphism $\sigma$ is injective on $\cL(x)$ (which is decidable, see~\cite{Ehrenfeucht&Rozenberg:1978}), we can take $d=1$ in Theorem~\ref{theo:maindetail} and the computation in the proof of Theorem~\ref{theo:main} gives the bound 
\[
	2 |\sigma|^{6(\#A)^2 + 6 |\sigma|^{28(\#A)^2}} + |\sigma|. 
\]

The notion of recognizability is also known as {\em circularity} in the terminology of {\em D0L-systems} \cite{Kari&Rozenberg&Salomaa:1997}.
Assume that $\sigma:A^* \to A^*$ is non-erasing and that $a \in A$ is a letter such that the language $\mathrm{Fac}(\sigma,a)$ defined as the set of factors occuring in $\sigma^n(a)$ for some $n$ is infinite.
Given a word $u = u_1 \cdots u_{|u|} \in \cL(a)$, we say that a triplet $(p,v,s)$ is an {\em interpretation} of $u$ if $\sigma(v) = pus$.
Two interpretations $(p,v,s),(p',v',s')$ are said to be {\em synchronized at position $k$} if there exist $i,j$ such that $1 \leq i \leq |v|$, $1 \leq j \leq |v'|$ and
\[
	\sigma(v_1 \cdots v_i) = p u_1 \cdots u_k 
	\quad \text{and} \quad
	\sigma(v_1' \cdots v_j') = p' u_1 \cdots u_k. 
\]
The word $u$ has a {\em synchronizing point (at position $k$)} if all its interpretations are synchronized (at position $k$).
The pair $(\sigma,a)$ is said to be {\em circular} if $\sigma$ is injective on $\mathrm{Fac}(\sigma,a)$ and if there is a constant $C$, called the {\em synchronizing delay} of $\sigma$, such that any word of length at least $C$ has a synchronizing point.
Thus, despite some considerations about whether we deal with fixed points or languages, recognizability and circularity are roughly the same notion and the synchronizing delay $C$ is associated with the constant of recognizablity $L$ through the equation $C = 2L+1$.
Using the termininology of D0L-systems, Klouda and Medkov\'a obtained the following result which greatly improves our bounds, but for restricted cases.
%
\begin{theorem}[\cite{Klouda&Medkova:2016}]
If $\# A = 2$ and if $(\sigma,a)$ is circular with $\sigma:A^* \to A^*$ a $k$-uniform morphism for some $k \geq 2$, then the synchronizing delay $C$ of $(\sigma,a)$ is bounded as follows:
\begin{enumerate}
\item
$C \leq 8$ if $k = 2$,
\item
$C \leq k^2+3k-4$ if $k$ is an odd prime number,
\item
$C \leq k^2 \left( \frac{k}{d}-1\right) +5k-4$ otherwise,
\end{enumerate}
where $d$ is the least divisor of $k$ greater than 1.
\end{theorem}

\section{Proof of Theorem~\ref{theo:maindetail}} \label{section:detailed result}

Like in Mossé's original proof, the proof of Theorem~\ref{theo:maindetail} goes in two steps. 

As a first step, we express the constant $M$ of Theorem~\ref{theo:mosse} in terms of the constants $N$, $R$, $k$ and $Q$ of Theorem~\ref{theo:maindetail}.
This is done in Proposition \ref{prop:rec1} with a proof following the lines of the proof of~\cite[Proposition 4.35]{Kurka:2003}.
The difference is that we take care of all the needed bounds to express the constant of recognizability.

As a second step, we show that the constant $L$ of Theorem~\ref{theo:mosse} can be taken equal to $M' + |\sigma^{d}|$, where $d$ is as defined in Theorem~\ref{theo:maindetail} and $M'$ is such that for all $i,m \in \Z$,
\[
	x_{[f^{(d)}(i)-M',f^{(d)}(i)+M']} =  x_{[m-M,m+M]} \Longrightarrow m \in E (x, \sigma^d).  
\]
We first start with the following lemma.

\begin{lemma}
\label{lemme:longueurs_u_v_REC}
Let $\sigma:A^* \to A^*$ be a non-erasing morphism, $u \in A^*$ be a word and $n$ be a positive integer.
If $v=v_0 \cdots v_{t+1}  \in A^*$ is a word of length $t+2$ such that $\sigma^n(v[1,t])$ is a factor of $\sigma^n(u)$, and $\sigma^n(u)$ is a factor of $\sigma^n(v)$, then
\[
	  \frac{\langle \sigma^n\rangle}{|\sigma^n|} |u|-2 \leq t \leq \frac{|\sigma^n|}{\langle \sigma^n\rangle} |u|.
\]
\end{lemma}

\begin{proof}
Indeed, since $\sigma^n(v{[1,t]})$ is a factor of $\sigma^n(u)$ we have
$t \langle \sigma^n \rangle \leq |\sigma^n(v{[1,t]})| \leq |\sigma^n(u)| \leq |u| |\sigma^n|$.
Hence $t \leq |u| |\sigma^n|/\langle \sigma^n\rangle$.
Similarly, since $\sigma^n(u)$ is a factor of $\sigma^n(v)$, we have $|u| \leq (t+2) |\sigma^n|/\langle \sigma^n\rangle$. 
We thus have
\[
	|u|\frac{\langle \sigma^n\rangle}{|\sigma^n|}-2 \leq t \leq |u| \frac{|\sigma^n|}{\langle \sigma^n\rangle} .
\]
\end{proof}

\begin{proposition}
\label{prop:rec1}
Let $\sigma :A^* \to A^*$ be a morphism with an admissible fixed point $x \in A^\Z$.
Assuming that $x$ is $k$-power-free and that there is some constant $N$ such that for all $n \in \mathbb{N}$, $|\sigma^n| \leq N \langle \sigma^n \rangle$,
we consider the constants
\begin{itemize}
\item
$R = \lceil N^2(k+1)+2N\rceil$; 
\item
$Q = 1 + p_x(R)\left(\sum_{\frac{R}{N} \leq i \leq RN+2} p_x(i) \right)$.
\end{itemize}
The constant $M = R|\sigma^Q|$ is such that for all $i,m \in \Z$,
\begin{align}
\label{align:recM}
x_{[f^{(1)}(i)-M,f^{(1)}(i)+M]} =  x_{[m-M,m+M]} \Longrightarrow m \in E (x, \sigma ).  
\end{align}

\end{proposition}

\begin{proof}
We follow the lines of the proof of Theorem~\ref{theo:mosse} that is in~\cite{Kurka:2003}.
Obviously, if $l$ satisfies~\eqref{align:recM}, then so does $l'$ whenever $l' \geq l$.
Let us show that such an $l$, with  $R|\sigma^Q|$, exists.

We proceed by contradiction, assuming that for all $l$, there exist $i,j$ such that $x_{[i-l,i+l]} = x_{[j-l,j+l]}$ with $i \in E(x, \sigma)$ and $j \notin E (x,\sigma )$.
For any integer $p$ such that $0 < p \leq Q$, we consider the integer $l_p=R|\sigma^p|$.
Let $i_p$ and $j_p$ be some integers such that 
\begin{equation*}
	x_{[i_p-l_p,i_p+l_p]} = x_{[j_p-l_p,j_p+l_p]} \quad  \text{with } i_p \in E(x,\sigma ) \text{ and} \, j_p \notin E(x,\sigma ).
\end{equation*}
We let $r_p$ and $s_p$ denote the smallest integers such that
\begin{align*}
&	\card \left( [i_p - r_p,i_p[\, \cap E(x,\sigma^p )\right)  =  \left\lceil\frac{R}{2}\right\rceil 
	\quad \text{and} \\
&	\card \left( [i_p,i_p + s_p] \cap E(x,\sigma^p ) \right)  =  \left\lfloor\frac{R}{2}\right\rfloor+1.	
\end{align*}
There is an integer $i_p'$ such that 
\[
	f^{(p)}(i_p') = i_p-r_p
	\quad \text{and} \quad
	f^{(p)}(i_p'+R) = i_p+s_p.
\]
We set 
\[
	u_p = x_{[i_p',i_p'+R[}.
\] 
We have $\sigma^p(u_p) = x_{[i_p-r_p,i_p+s_p[}$. 

Notice that any interval of length $l_p$ contains at least $R-1$ elements of $E (x ,  \sigma^p )$.
We thus have $i_p - l_p  \leq i_p - r_p \leq    i_p + s_p \leq i_p + l_p$. 
Consequently we also have 
\begin{align}
\label{align:upq}
x_{[j_p-r_p,j_p+s_p[} = \sigma^p(u_p).
\end{align}
However $j_p - r_p$ does not need to belong to $E(x, \sigma^p )$.
Let $j_p'$ and $t_p$ denote the unique integers such that
\begin{equation}
\label{eq:bornes jp'}
\begin{array}{rcl}
	\f^{(p)}(j_p') < & j_p - r_p & \leq \f^{(p)}(j_p'+1)	;	\\
	\f^{(p)}(j_p'+t_p+1) \leq & j_p + s_p & < \f^{(p)}(j_p'+t_p+2).
\end{array}
\end{equation}
Consequently $\sigma^p (x_{[j_p'+1,j_p'+t_p]})$ is a factor of $\sigma^p (u_p)$ and $\sigma^p (u_p)$ is a factor of $\sigma^p (x_{[j_p',j_p'+t_p+1]})$.
By Lemma~\ref{lemme:longueurs_u_v_REC}, we have 
\begin{equation}\label{eq:bornes tp}
	R\frac{\langle \sigma^p \rangle}{|\sigma^p|}-2 \leq t_p \leq R \frac{|\sigma^p|}{\langle \sigma^p \rangle}.
\end{equation}
Hence
\[
	\frac{R}{N}-2 \leq t_p \leq R N.
\]
Let $v_p = x_{[j_p',j_p'+t_p+1]}$.
The number of possible pairs of words $(u_p,v_p)$ is at most
\[
	p_x(R) \left( \sum_{ \frac{R}{N} \leq i \leq {RN+2}} p_x(i)\right) 
	< Q.
\]  
Therefore, there exist $p$ and $q$ in $[1, Q]$ such that $p<q$ and $(u_p,v_p) = (u_q,v_q)$. 
In particular we also have $t_p=t_q$.
We write 
$$
t=t_p, \qquad u = u_p, \qquad v = v_p, \qquad \tilde{v} = x_{[j_p'+1,j_p'+t]}.
$$
Using the above notation we recall that we have
\begin{eqnarray}
	u	=&	x_{[i_p',i_p'+R[} &= x_{[i_q',i_q'+R[}, \\
\label{eq:W}
	v	=&	x_{[j_p',j_p'+t+1]} &= x_{[j_q',j_q'+t+1]}.	
\end{eqnarray}
Let $A_p$, $B_p$, $A_q$ and $B_q$ be the words 
\begin{eqnarray*}
	A_p &=& x_{[j_p-r_p,f^{(p)}(j_p'+1)[};		\\
	B_p &=& x_{[f^{(p)}(j_p'+t+1),j_p+s_p[};	\\
	A_q &=& x_{[j_q-r_q,f^{(q)}(j_q'+1)[};		\\
	B_q &=& x_{[f^{(q)}(j_q'+t+1),j_q+s_q[}.
\end{eqnarray*}
We thus have 
\begin{equation}
\label{eq:A_pB_p}
	x_{[j_p-r_p,j_p+s_p[} = A_p \sigma^p(\tilde{v} ) B_p
	\quad \text{and} \quad
	x_{[j_q-r_q,j_q+s_q[} = A_q \sigma^q(\tilde{v} ) B_q,
\end{equation}
with, using~\eqref{eq:bornes jp'},
\begin{equation}\label{eq:longueur Ap Bp}
\max \{|A_p|,|B_p|\} \leq |\sigma^p| 
\quad \text{and} \quad
\max \{|A_q|,|B_q|\} \leq |\sigma^q|. 
\end{equation}
From~\eqref{align:upq} and~\eqref{eq:A_pB_p}, we obtain
\[
	\sigma^{q-p}(A_p) \sigma^q(\tilde{v}) \sigma^{q-p}(B_p) = A_q \sigma^q(\tilde{v} ) B_q.
\]
We claim that 
\begin{equation}
\label{eq:claim}
A_q = \sigma^{q-p}(A_p) \quad \text{(and hence $B_q = \sigma^{q-p}(B_p)$)}.
\end{equation} 
If not, the word $\sigma^q( \tilde{v}  )$ has a prefix which is a power $w^r$ with $r=\left\lfloor\frac{|\sigma^q(\tilde{v})|}{||A_q| - |\sigma^{q-p}(A_p)||}\right\rfloor$.
Since, using~\eqref{eq:bornes tp} and~\eqref{eq:longueur Ap Bp}, 
\[
	|\sigma^q( \tilde{v})| \geq t \langle \sigma^q\rangle \geq \left(\frac{R}{N}-2\right)\langle \sigma^q\rangle
\quad \hbox{ and } \quad 
	||A_q| - |\sigma^{q-p}(A_p)|| \leq |\sigma^q|,
\]
we deduce from the choice of $R$ that $r \geq k+1$, which contradicts the definition of $k$.
We thus have $A_q = \sigma^{q-p}(A_p)$ and $B_q = \sigma^{q-p}(B_p)$.

We now show that 
\begin{equation}
\label{eq:claim end}
	[j_q-r_q,j_q+s_q] \cap E(x,\sigma) = ([i_q-r_q,i_q+s_q] \cap E(x,\sigma)) - (i_q-j_q).
\end{equation}
This will contradict the fact that $i_q$ belongs to $E(x,\sigma)$ and $j_q$ does not.

By~\eqref{eq:W}, we have
\[
	\sigma^p(v) 
	= x_{[f^{(p)}(j_p'),f^{(p)}(j_p'+t+2)[} 
	= x_{[f^{(p)}(j_q'),f^{(p)}(j_q'+t+2)[}.
\]
Since $\sigma^p(u)$ is a factor of $\sigma^p(v)$, we deduce from~\eqref{eq:bornes jp'} that there exists $m_q \in \Z$ such that 
\begin{align}
\label{align:qqq}
	f^{(p)}(j_q') < m_q -r_p < m_q + s_p < f^{(p)}(j_q'+t+2)
\end{align}
and 
\[
	x_{[m_q-r_p,m_q+s_p[} = \sigma^p(u) = A_p \sigma^p(\tilde{v}) B_p.
\]
By applying $\sigma^{q-p}$, we obtain 
\[
	x_{[f^{(q-p)}(m_q-r_p),f^{(q-p)}(m_q+s_p)[} = A_q \sigma^q(\tilde{v}) B_q ,
\]
and, from \eqref{align:qqq}, 
\[
	f^{(q)}(j_q') < f^{(q-p)}(m_q -r_p) < f^{(q-p)}(m_q + s_p) < f^{(q)}(j_q'+t+2) .
\]
As we also have 
\[
	x_{[j_q - r_q,j_q+s_q[} = A_q \sigma^q(\tilde{v}) B_q
\]
with, by~\eqref{eq:bornes jp'},
\[
	\f^{(q)}(j_q') <  j_q - r_q \leq \f^{(q)}(j_q'+1) \leq \f^{(q)}(j_q'+t+1) \leq j_q + s_q < \f^{(q)}(j_q'+t+2),
\]
we apply the same argument as to show~\eqref{eq:claim} and get $j_q-r_q = f^{(q-p)}(m_q-r_p)$ (hence $j_q+s_q = f^{(q-p)}(m_q+s_p)$).
We thus get that $j_q-r_q$ belongs to $E(x, \sigma^{q-p}) \subset E(x, \sigma)$.
Since we also have 
\begin{eqnarray*}
	x_{[{f^{(1)}}^{-1}(j_q-r_q),{f^{(1)}}^{-1}(j_q+s_q)[} = \sigma^{q-p-1}(x_{[m_q-r_p,m_q+s_p[}) = \sigma^{q-p-1}(A_p \sigma^p(\tilde{v}) B_p), \\
	x_{[{f^{(1)}}^{-1}(i_q-r_q),{f^{(1)}}^{-1}(i_q+s_q)[} = \sigma^{q-1}(x_{[i_q',i_q'+R[}) = \sigma^{q-p-1}(A_p \sigma^p(\tilde{v}) B_p),
\end{eqnarray*}
we get 
\[
	x_{[{f^{(1)}}^{-1}(j_q-r_q),{f^{(1)}}^{-1}(j_q+s_q)[} = x_{[{f^{(1)}}^{-1}(i_q-r_q),{f^{(1)}}^{-1}(i_q+s_q)[}
\]
with $j_q-r_q, i_q-r_q$ belonging to $E(x,\sigma)$.
By applying $\sigma$ to these two word, we thus obtain~\eqref{eq:claim end}, which ends the proof.
\end{proof}

In Proposition~\ref{prop:rec1}, we compute a constant such that any long enough word can be cut into words in $\sigma(A)$ in a unique way except for a prefix and a suffix of bounded length.
However it does not give information on the letters in $A$ that the words in $\sigma(A)$ come from.
A key argument in Mossé's original proof is to prove the existence of an integer $d$ such that for all $a,b \in A$, if $\sigma^n(a) = \sigma^n(b)$ for some $n$, then $\sigma^d(a)=\sigma^d(b)$.
We then prove that the constant $L$ of Theorem~\ref{theo:mosse} can be taken equal to $M + |\sigma^{d+1}|$, where $M$ is the constant of Proposition~\ref{prop:rec1} associated with $\sigma^{d+1}$.
Theorem~\ref{theo:exponent card A} below ensures that we can take $d = \# A-1$, which ends the proof of Theorem~\ref{theo:maindetail}.

\begin{theorem}[{\cite[Theorem 3]{Ehrenfeucht&Rozenberg:1978}}]
\label{theo:exponent card A}
Let $\sigma:A^* \to A^*$ be a morphism such that $\sigma(A) \neq \{\varepsilon\}$.
For any words $u,v \in A^*$, we have 
\[
	\sigma^{\# A-1}(u) \neq \sigma^{\# A-1}(v) \Rightarrow \forall n, \sigma^n(u) \neq \sigma^n(v). 
\] 
\end{theorem}

We give the proof of Mossé's second step result for the sake of completeness.

\begin{proposition}
\label{prop:rec2}
Let $\sigma :A^* \to A^*$ be a morphism with an admissible fixed point $x \in A^\Z$.
Let $d \in \{1,2,\dots,\# A\}$ be such that for any words $u,v \in \cL(x)$,  
\[
	\sigma^{d-1}(u) \neq \sigma^{d-1}(v) \Rightarrow \forall n, \sigma^n(u) \neq \sigma^n(v). 
\] 
If $M$ is a constant such that for all $i,m \in \Z$,
\[
x_{[f^{d}(i)-M,f^{d}(i)+M]} = x_{[m-M,m+M]} \Longrightarrow m \in E (x, \sigma^{d}),
\]
then $\sigma$ is recognizable on $x$ and its constant of recognizability for $x$ is at most $M + |\sigma^{d}|$.
\end{proposition}

\begin{proof}
Let $i,m \in \mathbb{Z}$ such that 
$$
	x_{[f^{(1)}(i)-M-|\sigma^d|,f^{(1)}(i)+M+|\sigma^d|]}
	=
	x_{[m-M-|\sigma^d|, m+M+|\sigma^d|]}.	
$$
By definition of $M$, there exists $j \in \mathbb{Z}$ such that $m = f^{(1)}(j)$.
Our goal is to show that $x_i = x_j$.

There exists $k \in \mathbb{Z}$ such that 
\[
	f^{(1)}(i)-|\sigma^d | < f^{(d)}(k) \leq f^{(1)}(i) < f^{(d)}(k+1) \leq f^{(1)}(i)+|\sigma^d |. 
\]
In particular, this implies that $f^{(d-1)}(k) \leq i < f^{(d-1)}(k+1)$.

Consider $c = f^{(1)}(i) - f^{(d)}(k)$ and $d = f^{(d)}(k+1) - f^{(1)}(i)$.
We have 
\begin{eqnarray*}
	x_{[f^{(d)}(k)-M,f^{(d)}(k)+M]} 
	&=& 
	x_{[f^{(1)}(j)-c-M,f^{(1)}(j)-c+M]};	\\
	x_{[f^{(d)}(k+1)-M,f^{(d)}(k+1)+M]} 
	&=& 
	x_{[f^{(1)}(j)+d-M,f^{(1)}(j)+d+M]}.
\end{eqnarray*}
By definition of $M$, there exists $l \in \mathbb{Z}$ such that 
\[
	f^{(d)}(l) = f^{(1)}(j)-c
	\quad \text{and} \quad
	f^{(d)}(l+1) = f^{(1)}(j)+d.
\]
We thus have $f^{(d-1)}(l) \leq j < f^{(d-1)}(l+1)$, and,
\[
	x_{[f^{(d)}(k),f^{(d)}(k+1)[} = x_{[f^{(d)}(l),f^{(d)}(l+1)[}.
\]
Hence $\sigma^d(x_k) = \sigma^d(x_l)$.
By definition of $d$, we also have $\sigma^{d-1}(x_k) = \sigma^{d-1}(x_l)$.
Hence 
\[
	x_{[f^{(d-1)}(k),f^{(d-1)}(k+1)[} = x_{[f^{(d-1)}(l),f^{(d-1)}(l+1)[}.
\]
Since we have $f^{(1)}(i) - f^{(d)}(k) = f^{(1)}(j) - f^{(d)}(l)$, we also have $i - f^{(d-1)}(k) = j - f^{(d-1)}(l)$.
Hence $x_i = x_j$.
\end{proof}

\section{Proof of Theorem~\ref{theo:main}}
\label{section: primitive case}

In this section, we show that the constants appearing in Theorem~\ref{theo:maindetail} can all be bounded by some computable constants.
In all what follows, we assume that $\sigma:A^* \to A^*$ is a primitive morphism. 
By taking a power of $\sigma$ if needed, we can assume that it has an admissible fixed point $x \in A^\Z$.
Furthermore, we have $\cL(x) = \cL(y)$ for all admissible fixed points $y$ of $\sigma$. 
We let $\cL(\sigma)$ denote this set.
The constants appearing in Theorem~\ref{theo:maindetail} are thus the same whatever the admissible fixed point we consider and the morphism is recognizable.

With the morphism $\sigma$, one associates its {\em incidence matrix} $M_\sigma$ defined by $(M_\sigma)_{a,b} = |\sigma(b)|_a$, where $|u|_a$ denotes the number of occurrences of the letter $a$ in the word $u$.

\begin{lemma}[\cite{Horn&Johnson:1990}]
\label{lemma:wielandt}
A $d \times d$ matrix $M$ is primitive if, and only if, there is an integer $k \leq d^2 -2d +2$ such that $M^k$ contains only positive entries. 
\end{lemma}

Given an infinite word $x \in A^\Z$ and a word $u \in \cL(x)$, a {\em return word} to $u$ in $x$ is a word $r$ such that $ru$ belongs to $\cL(x)$, $u$ is a prefix of $ru$ and $ru$ contains exactly two occurrences of $u$.
The infinite word $x$ is {\em linearly recurrent} if it is {\em recurrent} (all words in $\mathcal{L} (x)$ appear infinitely many times in $x$) and there exists some constant $K$ such that for all $u \in \cL(x)$, any return word to $u$ has length at most $K|u|$.
The set of return words to $u$ in $x$ is denoted $\cR_{x,u}$.

The next two results give bounds on the constants appearing in Theorem~\ref{theo:maindetail}.

\begin{theorem}[\cite{Durand&Host&Skau:1999}]
\label{theo:encad}
If $x \in A^\Z$ is a aperiodic and linearly recurrent sequence (with constant $K$), then $x$ is $(K+1)$-power-free and $p_x(n) \leq K n$ for all $n$.
\end{theorem}

\begin{proposition}[\cite{Durand:1998}]
\label{prop:sublinrec}
Let $\sigma : A^* \to A^*$ be an aperiodic primitive morphism and $x$ be one of its admissible fixed points. 
Then we have 
\[
	|\sigma^n| \leq |\sigma|^{(\#A)^2} \langle \sigma^n\rangle
	\qquad \text{for all } n
\]
and $x$ is linearly recurrent for some constant 
\[
	K_\sigma < |\sigma|^{4(\#A)^2}.
\]
\end{proposition}
\begin{proof}
Durand \cite{Durand:1998} showed that the constant of linear recurrence $K_\sigma$ is at most equal to $RN|\sigma|$, where 
\begin{itemize}
\item
$N$ is a constant such that $|\sigma^n| \leq N \langle \sigma^n\rangle$ for all $n$;
\item
$R$ is the maximal length of a return word to a word of length $2$ in $\mathcal{L}(\sigma)$.
\end{itemize}
We only prove here that $N \leq |\sigma|^{(\#A)^2}$ and $R \leq 2|\sigma|^{2(\#A)^2}$. The constant of linear recurrence is thus at most $2 |\sigma|^{1+3(\# A)^2} < |\sigma|^{4(\#A)^2}$.

Let us write $d = \#A$.
By Lemma~\ref{lemma:wielandt}, the matrix $M_\sigma^{d^2}$ contains only positive entries.
For all $n \geq 0$ and all $a \in A$, we have $|\sigma^{n+d^2}(a)| = \sum_{b \in A} |\sigma^{d^2}(a)|_b|\sigma^n(b)| \geq |\sigma^n|$.
Since this is true for all $a$, we get $|\sigma^n| \leq \langle \sigma^{n+{d^2}}\rangle \leq |\sigma^{d^2}| \langle \sigma^n\rangle$, so $N \leq |\sigma^{d^2}|$.

Let $a \in A$ such that $\sigma$ is prolongable on $a$.
Thus for all $n$, any word that occurs in $\sigma^n(a)$ also occurs in $\sigma^{n+1}(a)$.
Let us show that for all $n > d^2$, any word $u \in \mathcal{L}(\sigma)$ of length $2$ occurs in $\sigma^n(a)$.
For all $n$, the words of length 2 that occur in $\sigma^{n+1}(a)$ occurs in images under $\sigma$ of the words of length 2 that occur in $\sigma^n(a)$. 
As any word occurring in $\sigma^n(a)$ also occurs in $\sigma^{n+1}(a)$, the words of length 2 that occurs in $\sigma^{n+1}(a)$ are those that occur in $\sigma^n(a)$ together with those occurring in the images under $\sigma$ of these words.
Thus, if there is a word of length 2 that does not occur in $\sigma^n(a)$, there is a sequence $(u_1,u_2,\dots,u_n)$ of words of length $2$ in $\mathcal{L}(\sigma)$ such that for all $i \leq n$, $u_i$ occurs in $\sigma^i(a)$ and does not occur in $\sigma^{i-1}(a)$.
Hence all words $u_1, \dots, u_n$ are distinct.
For $n > d^2$, this is a contradiction since there are at most $d^2$ words of length $2$ on the alphabet $A$.
Thus, for any letter $b \in A$, all words $u \in \mathcal{L}(\sigma)$ of length $2$ occur in $\sigma^{2d^2}(b)$.
We deduce that $R \leq 2|\sigma^{2d^2}|$.  
\end{proof}

\begin{proof}[Proof of Theorem~\ref{theo:main}]
We just have to make the computation.
Using Theorem~\ref{theo:encad}, Proposition~\ref{prop:sublinrec} and the notation of Theorem \ref{theo:maindetail}, we can take $d = \# A$ and we successively have 
\begin{eqnarray*}
	k 	&\leq & 1 + K_\sigma 	 \leq  |\sigma|^{4d^2} 	,
	\\
	N 	& \leq & |\sigma|^{d^2},
	\\
	R    & = & \lceil  N^2(k+1)+2N\rceil 	 \leq  |\sigma|^{2d^2} (|\sigma|^{4d^2}+1) + 2|\sigma|^{d^2}   \leq  2|\sigma|^{6d^2}	,
	\\
	Q 	&=& 1 + p_x(R) ,
	\left(
	\sum_{\frac{R}{N} \leq i \leq RN+2} p_x(i) 
	\right)  \leq  K_\sigma 2|\sigma|^{6d^2}
	\left(
	\sum_{0 \leq i \leq 2+2|\sigma|^{7d^2}} i K_\sigma
	\right)  \leq  6 |\sigma|^{28d^2}
\end{eqnarray*}
We finally get that the constant of recognizability of $\sigma$ is at most 
\[
	2 |\sigma|^{6d^2} |\sigma|^{6 d |\sigma|^{28d^2}} + |\sigma|^d 
	= 
	2 |\sigma|^{6d^2 + 6d |\sigma|^{28d^2}} + |\sigma|^d. 
\]
\end{proof}

\def\ocirc#1{\ifmmode\setbox0=\hbox{$#1$}\dimen0=\ht0 \advance\dimen0
  by1pt\rlap{\hbox to\wd0{\hss\raise\dimen0
  \hbox{\hskip.2em$\scriptscriptstyle\circ$}\hss}}#1\else {\accent"17 #1}\fi}

%

\end{document}